\documentclass[letterpaper, 10pt, conference]{ieeeconf}      

\IEEEoverridecommandlockouts                              
\usepackage{graphics} 
\usepackage{epsfig} 
\usepackage{mathptmx} 
\usepackage{times} 
\usepackage{amsmath} 
\usepackage{amssymb}  
\usepackage{subfigure}
\usepackage{verbatim}
\usepackage{cite}

\newtheorem{theorem}{Theorem}

\usepackage{color}

\newtheorem{definition}{Definition}

\newtheorem{lemma}{Lemma}

{ 
\newtheorem{remark}{Remark}

}
\newcounter{prob1}
\newcounter{prob2}
\newcounter{prob3}
\newcounter{prob4}
\newcounter{prob5}

\title{\LARGE \bf Bounds on the Achievable Rate of Noisy feedback Gaussian Channels under Linear Feedback Coding Scheme}

\author{ \parbox{5 in}{\centering Chong Li and Nicola Elia
         \thanks{This work was supported by NSF under grant number ECS-0901846}\\
         Department of Electrical and Computer Engineering, Iowa State University\\
         Ames, IA, 50011\\
         Email: $\lbrace$chongli, nelia$\rbrace$@iastate.edu\\}
}

\begin{document}

\maketitle \thispagestyle{empty} \pagestyle{empty}
\begin{abstract}
In this paper, we investigate the additive Gaussian noise channel with noisy feedback. We consider the setup of linear coding of the feedback information and Gaussian signaling of the message (i.e. Cover-Pombra Scheme). Then, we derive the upper and lower bounds on the largest achievable rate for this setup. We show that these two bounds can be obtained by solving two convex optimization problems. Finally, we present some simulations and discussion.
\end{abstract}


\setcounter{prob1}{1}
\setcounter{prob2}{2}
\setcounter{prob3}{3}
\setcounter{prob4}{4}
\setcounter{prob5}{5}
\section{Introduction}
\indent The study of the additive Gaussian noise channel with feedback has been a hot research topic for decades. So far, a large body of work has looked at the ideal feedback case and obtained many notable results \cite{Schalkwijk66,Omura68,Massey1990,cover89,Kim10}. As an illustration, it is known that noiseless feedback improves the error exponent and reduces the coding complexity. However, only few papers have studied channels with noisy feedback and many open problems still exist. Literature on noisy feedback problems can be largely classified into two categories. The first category studies the usefulness of noisy feedback by investigating reliability functions and error exponents \cite{Kim07},\cite{Love_arxiv}. The second focuses on the derivation of coding schemes based on the well-known Schalkwijk-Kailath scheme. We refer interested readers to \cite{Omura68,Lavenberg71,Martins08,Chance10,Kumar} for details.\\
\indent In this paper, we investigate the behavior of the largest achievable rate of the additive Gaussian noise channel with noisy feedback, under the restriction of linear feedback coding scheme (i.e.Cover-Pombra Scheme). We derive informative upper and lower bounds on the largest achievable rate. This upper bound outperforms the bound presented in \cite{Love10}, especially, in the case of having small feedback noise. The lower bound shows the enhancement, in terms of the achievable rate, by exploiting the noisy feedback link. Additionally, the derived bounds provide insight on how the noisy feedback channel behaves with respect to the feedback noise.\\
\indent The paper is organized as follows. In Section \ref{sec_pre}, we introduce some important definitions and lemmas, which are used throughout the paper. In Section III, we introduce the signal model of the noisy feedback channel and then derive the formula of the achievable rate. In Section \ref{sec_upper} and \ref{sec_lower}, we derive an upper bound and a lower bound on the largest achievable rate, respectively. We present some simulation results and discussion in Section \ref{sec_simu} and conclude the paper in Section \ref{sec_conclude}.\\
\indent Notations: Uppercase and corresponding lowercase letters $(e.g. Y,Z,y,z)$ denotes random variables and realizations, respectively. $x^n$ represents the vector $[x_1,x_2,\cdots,x_n]^T$ and $x^0=\emptyset$. $\mathbf{I}_n$ represents an $n\times n$ identity matrix. $\mathbf{K}_n>0$ ($\mathbf{K}_n\geq 0$) denotes that the $n\times n$ matrix $\mathbf{K}_n$ is positive definite (semi-definite). $log$ denotes the logarithm base $2$ and $0\log0=0$. The expectation operator over $X$ is presented as $\mathbb{E}(X)$.

\section{Technical Preliminaries}
\indent In this section, we review some main definitions and Lemmas in information theory.
\begin{definition}\cite{cover06}
The mutual information $I(X;Y)$ between two random variables with joint density $f(x,y)$ is defined as
\begin{equation*}
I(X;Y)=\int f(x,y)\log\frac{f(x,y)}{f(x)f(y)}dxdy
\end{equation*}
\end{definition}

\indent Let $h(X)$ denote the differential entropy of a random variable X. Then it is clear that
\begin{equation*}
I(X;Y)=h(Y)-h(Y|X)
\end{equation*}

\indent Next, we present a useful Lemma as follows.
\begin{lemma}\cite{cover06}
Let the random vector $X\in \mathbb{R}^{n}$ have zero mean and covariance $\mathbf{K}_{x,n}=\mathbb{E}XX^T$ (i.e. $\mathbf{K}_{ij}=\mathbb{E}X_iX_j$, $1\leq i,j\leq n$). Then
\begin{equation*}
h(X)\leq \frac{1}{2}\log(2\pi e)^n\det\mathbf{K}_{x,n}
\end{equation*}
with equality if and only if $X\sim \mathit{N}(0,\mathbf{K}_{x,n})$.
\label{lem_pre01}
\end{lemma}

\indent Finally, we introduce an meaningful notion of directivity to the information flow through a channel \cite{Massey1990}.
\begin{definition}
The directed information $I(X^n\rightarrow Y^n)$ from a sequence $X^n$ to a sequence $Y^n$ is defined by
\begin{equation*}
I(X^n\rightarrow Y^n)=\sum_{i=1}^{n}I(X^i;Y_{i}|Y^{i-1}).
\end{equation*}
\end{definition}

\indent It has been shown in \cite{Massey1990}, when feedback is present, directed information is a more useful quantity than the traditional mutual information.

\label{sec_pre}

\section{Modeling and Achievable Rate}

\begin{figure}
\includegraphics[scale=0.4]{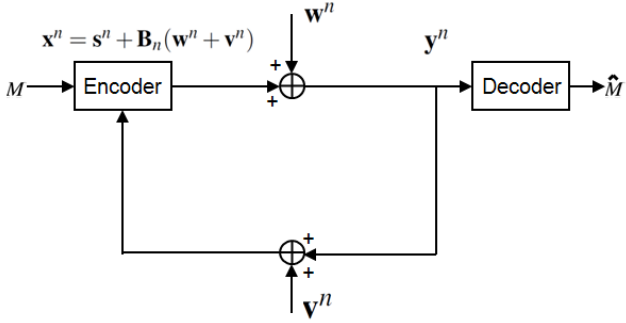}
\caption{A additive Gaussian noise channel with noisy feedback}
\label{ISIT_scheme01}
\end{figure}

\subsection{Modeling}
\indent Consider a point-to-point additive Gaussian noise channel with access to an additive Gaussian noise feedback link. Since it is difficult to characterize the capacity in general, we consider a system with linear encoding of the feedback signal and Gaussian signaling of the message (i.e. shown in a vector form in Fig. \ref{ISIT_scheme01})\cite{Love10}. \\

\vspace{2mm}

\indent \textit{The channel input signal:}  $\mathbf{x}^n=\mathbf{s}^n+\mathbf{B}_n(\mathbf{w}^n+\mathbf{v}^n)$\\ \indent \textit{The channel output signal:} $\mathbf{y}^n=\mathbf{s}^n+\mathbf{B}_n(\mathbf{w}^n+\mathbf{v}^n)+\mathbf{w}^n$\\
\indent \textit{The power constraint:} $tr(\mathbf{K}_{s,n}+\mathbf{B}_n(\mathbf{K}_{w,n}+\mathbf{K}_{v,n})\mathbf{B}_n^T)\leq nP$\\
\vspace{1mm}

where $\mathbf{w}^n\sim \mathit{N}(0,\mathbf{K}_{w,n})$ and $\mathbf{v}^n\sim \mathit{N}(0,\mathbf{K}_{v,n})$. We assume that $\mathbf{K}_{w,n}>0$ and $\mathbf{K}_{v,n}>0$ and, therefore, these covariance matrices are invertible. Here, $\mathbf{s}^n\sim \mathit{N}(0,\mathbf{K}_{s,n})$ is the message information vector. $\mathbf{B}_n$ is an $n\times n$ strictly lower triangular linear encoding matrix. Note that the one-step delay in the feedback link is dealt with the structure of matrix $\mathbf{B}_n$ and random variables $\mathbf{s}^n$,$\mathbf{v}^n$,$\mathbf{w}^n$ are automatically assumed to be independent.\\
\begin{remark}
The bounds provided later in the paper are only valid for this specific setup. In other words, the bounds may not be true for the capacity (i.e. the maximum achievable rate over all encoding strategies). Informative computable bounds on the capacity of channels with noisy feedback, however, are not known.
\end{remark}

\subsection{n-block Achievable Rate}
\indent Based on the above model, we obtain the n-block achievable rate $R_n^{noisy}$ as
\begin{equation*}
\begin{split}
R_n^{noisy}=&\frac{1}{n}I(M;Y^n)\\
=&\frac{1}{n}I(S^n;Y^n)  \qquad (a)\\
=&\frac{1}{n}(h(Y^n)-h(Y^n|S^n))\\
=&\frac{1}{n}\log{\frac{\det{((\mathbf{I}_n+\mathbf{B}_n)\mathbf{K}_{w,n}(\mathbf{I}_n+\mathbf{B}_n)^T+\mathbf{B}_n\mathbf{K}_{v,n}\mathbf{B}_n^T+\mathbf{K}_{s,n})}}{\det{((\mathbf{I}_n+\mathbf{B}_n)\mathbf{K}_{w,n}(\mathbf{I}_n+\mathbf{B}_n)^T+\mathbf{B}_n\mathbf{K}_{v,n}\mathbf{B}_n^T)}}}\\
\end{split}
\end{equation*}
(a) follows the fact that the message information vector $S^n$ is determined by the message $M$ and the last equality follows from Lemma \ref{lem_pre01}. We denote the largest n-block achievable rate under power constraint $P$ as $R_{n,max}^{noisy}(P)$ where
\begin{equation}
R_{n,max}^{noisy}(P)=\max_{\mathbf{B}_n,\mathbf{K}_{s,n}}R_n^{noisy}\\
\label{eq1_1}
\end{equation}

\indent If the feedback is ideal (i.e. $\mathbf{K}_{v,n}=\mathbf{0}_{n}$), expression (\ref{eq1_1}) is simplified to
\begin{equation*}
C_n^{fb}(P)=\max_{\mathbf{B}_n,\mathbf{K}_{s,n}}\frac{1}{n}\log{\frac{\det{((\mathbf{I}_n+\mathbf{B}_n)\mathbf{K}_{w,n}(\mathbf{I}_n+\mathbf{B}_n)^T+\mathbf{K}_{s,n})}}{\det\mathbf{K}_{w,n}}}\\
\end{equation*}
where $C_n^{fb}(P)$ denotes the $n$-block capacity of the additive Gaussian noise channel with ideal feedback and $tr(\mathbf{K}_{s,n}+\mathbf{B}_n\mathbf{K}_{w,n}\mathbf{B}_n^T)\leq nP$. Note that the fact $\det(\mathbf{I}_n+\mathbf{B}_n)=\det(\mathbf{I}_n+\mathbf{B}_n)^T=1$ is applied in the simplification. \cite{Boyed98} found that this problem can be transformed into a well-known convex optimization problem called the matrix determinant maximization (max-det) problem and, then, solved efficiently.\\
\indent If there exists no feedback (i.e. $\mathbf{B}_{n}=\mathbf{0}_{n}$), expression (\ref{eq1_1}) is simplified to
\begin{equation*}
C_n^{open}(P)=\max_{\mathbf{K}_{s,n}}\frac{1}{n}\log{\frac{\det(\mathbf{K}_{w,n}+\mathbf{K}_{s,n})}{\det\mathbf{K}_{w,n}}}\\
\end{equation*}
where $C_n^{open}(P)$ denotes the $n$-block capacity of the open-loop additive Gaussian noise channel and $tr(\mathbf{K}_{s,n})\leq nP$. It is well known that this optimization problem can be solved by water-filling on the eigenvalues of $\mathbf{K}_{w,n}^{-1}$.

\indent Unlike the above two simplified forms, the optimization problem (\ref{eq1_1}) is difficult to solve (i.e. not easily expressed in a convex form). Furthermore, the expression for $R_n^{noisy}$ does not provide much information about the behavior of the channel with respect to the feedback noise. This motivates us to derive the effective upper and lower bounds on $R_{n,max}^{noisy}(P)$, from which we can discover some characterizations of the noisy feedback channel.   
\label{sec_m}

\section{An Upper Bound on $R_{n,max}^{noisy}(P)$}
\indent In this section, we first present an upper bound on the achievable rate of a general channel with additive noise feedback (i.e. not restricted to be an additive Gaussian noise channel). Then, we show that, for the noisy feedback Gaussian channel under linear feedback coding scheme, this upper bound can be obtained by solving a convex optimization problem.
\begin{lemma}
For a point-to-point communication channel with additive noise feedback, we have
\begin{equation*}
I(M;Y^n)\leq I(X^n\rightarrow Y^n|V^n)\leq I(X^n\rightarrow Y^n)\leq I(X^n;Y^n)
\end{equation*}
The first and second equalities hold if there exists an ideal feedback (i.e. $V^n=0$). The last equality holds if these exists no feedback.
\label{lem01}
\end{lemma}
\begin{proof}
\begin{equation*}
\begin{split}
&I(M;Y^n)\\
=&h(M)-h(M|Y^n)\\
\leq &h(M)-h(M|Y^n,V^n)\\
\stackrel{(a)}{=}&h(M|V^n)-h(M|Y^n,V^n)\\
=&I(M;Y^n|V^n)\\
=&h(Y^n|V^n)-h(Y^n|M,V^n)\\
\end{split}
\end{equation*}
\begin{equation*}
\begin{split}
=&\sum_{i=1}^n h(Y_i|Y^{i-1},V^n)-h(Y_i|Y^{i-1},M,V^n)\\
\stackrel{(b)}{=}&\sum_{i=1}^n h(Y_i|Y^{i-1},V^n)-h(Y_i|Y^{i-1},M,V^n,X^i) \\
\stackrel{(c)}{=}&\sum_{i=1}^n h(Y_i|Y^{i-1},V^n)-h(Y_i|Y^{i-1},X^i, V^n) \\
=&\sum_{i=1}^n I(X^i;Y_i|Y^{i-1},V^n)\\
=&I(X^n\rightarrow Y^n|V^n)\\
\end{split}
\end{equation*}
(a) follows from the fact that $M$ and $V^n$ are independent. (b) follows from the fact that $X^i$ can be determined by $M$ and the outputs of the feedback link (i.e. $Y^{i-1}+V^{i-1}$). (c) follows from the Markov chain $M\rightarrow (Y^{i-1},X^i,V^n) \rightarrow Y_i$. Note that the equality holds if $V^n=0$.\\
\indent Next, we have
\begin{equation*}
\begin{split}
&I(X^n\rightarrow Y^n|V^n)\\
\stackrel{(d)}{=}&\sum_{i=1}^n h(Y_i|Y^{i-1},V^n)-h(Y_i|Y^{i-1},X^i, V^n) \\
\stackrel{(e)}{=}&\sum_{i=1}^n h(Y_i|Y^{i-1},V^n)-h(Y_i|Y^{i-1},X^i) \\
\leq &\sum_{i=1}^n h(Y_i|Y^{i-1})-h(Y_i|Y^{i-1},X^i)\\
=&\sum_{i=1}^n I(X^i;Y_i|Y^{i-1})\\
=&I(X^n\rightarrow Y^n)\\
\end{split}
\end{equation*}
where (d) follows from step (c) and (e) follows from the Markov chain $V^n \rightarrow (Y^{i-1},X^i) \rightarrow Y_i$. Note that the equality holds if $V^n=0$.\\
\indent The last inequality $I(X^n\rightarrow Y^n)\leq I(X^n;Y^n)$ is proved in \cite{Massey1990}.
\end{proof}

\indent It is known that, for ideal feedback channels, the directed information $I(X^n\rightarrow Y^n)$ is an appropriate measure on the achievable rate and, therefore, can correctly characterize the ideal feedback channel capacity \cite{Tati09}. However, Lemma \ref{lem01} shows that, for noisy feedback channels, the conditional directed information $I(X^n\rightarrow Y^n|V^n)$ performs as a better upper bound on the achievable rate than $I(X^n\rightarrow Y^n)$. This motivates us to take $I(X^n\rightarrow Y^n|V^n)$ as an upper bound on $C_n^{noisy}(P)$ instead of $I(X^n\rightarrow Y^n)$ and investigate the following optimization problem.
\begin{equation}
\begin{split}
\quad \underset{\mathbf{B}_n,\mathbf{K}_{s,n}}{\rm maximize} &\quad \frac{1}{n}I(X^n\rightarrow Y^n|V^n)\\
\text{subject to} &\quad tr(\mathbf{K}_{x,n})\leq nP, \quad \mathbf{K}_{s,n}\geq 0 \\
&\quad \text{$\mathbf{B}_n$ is strictly lower triangular}\\
\end{split}
\label{eq2_01}
\end{equation}

\indent Next, we show that the above optimization problem can be transformed into a convex form.
\begin{theorem}
An upper bound on the largest n-block achievable rate of linear feedback coding scheme for Gaussian channels with additive noise feedback, as shown in Fig. \ref{ISIT_scheme01}, can be obtained as the optimal value of the following convex optimization problem.
\begin{equation*}
\begin{split}
\quad \underset{\mathbf{H}_n,\mathbf{B}_{n}}{\rm maximize} &\quad \frac{1}{2n}\log\det \begin{bmatrix} \mathbf{K}_{v,n}^{-1} & \mathbf{B}_n^T \\ \mathbf{B}_n & \mathbf{H}_n \end{bmatrix}-\frac{1}{2n}\log\det(\mathbf{K}_{v,n}^{-1}\mathbf{K}_{w,n})\\
\text{subject to} &\quad tr(\mathbf{H}_n-\mathbf{K}_{w,n}\mathbf{B}_n^T-\mathbf{B}_n\mathbf{K}_{w,n}-\mathbf{K}_{w,n})\leq nP \\
&\quad \begin{bmatrix}  \mathbf{H}_n & \mathbf{I}_n+\mathbf{B}_n^T & \mathbf{B}_n^T\\ \mathbf{I}_n+\mathbf{B}_n & \mathbf{K}_{w,n}^{-1}& \mathbf{0}_{n}\\ \mathbf{B}_n & \mathbf{0}_{n}& \mathbf{K}_{v,n}^{-1} \end{bmatrix}\geq 0\\
&\quad \text{$\mathbf{B}_n$ is strictly lower triangular}\\
\end{split}
\end{equation*}
\label{thm01}
\end{theorem}

\begin{proof}
\indent We are beginning with the optimization problem (\ref{eq2_01}). Let $\mathbf{H}_n=(\mathbf{I}_n+\mathbf{B}_n)\mathbf{K}_{w,n}(\mathbf{I}_n+\mathbf{B}_n)^T+\mathbf{K}_{s,n}+\mathbf{B}_n\mathbf{K}_{v,n}\mathbf{B}_n^T$, we have
\begin{equation*}
\begin{split}
&I(X^n\rightarrow Y^n|V^n)\\
\stackrel{(a)}{=}&\sum_{i=1}^n h(Y_i|Y^{i-1},V^n)-h(Y_i|Y^{i-1},X^i)\\
=&\sum_{i=1}^n h(Y_i|Y^{i-1},V^n)-h(X_i+W_i|Y^{i-1},X^i,W^{i-1})\\
=&\sum_{i=1}^n h(Y_i|Y^{i-1},V^n)-h(W_i|W^{i-1})\\
=&h(Y^n|V^n)-h(W^n)\\
=&h((\mathbf{I}_n+\mathbf{B}_n)W^n+\mathbf{B}_nV^n+S^n|V^n)-h(W^n)\\
=&h((\mathbf{I}_n+\mathbf{B}_n)W^n+S^n)-h(W^n)\\
\stackrel{(b)}{=}&\frac{1}{2}\log{\frac{\det{((\mathbf{I}_n+\mathbf{B}_n)\mathbf{K}_{w,n}(\mathbf{I}_n+\mathbf{B}_n)^T+\mathbf{K}_{s,n})}}{\det{\mathbf{K}_{w,n}}}}\\
=&\frac{1}{2}\log{\frac{\det{(\mathbf{H}_n-\mathbf{B}_n\mathbf{K}_{v,n}\mathbf{B}_n^T)}}{\det{\mathbf{K}_{w,n}}}}\\
\end{split}
\end{equation*}
where (a) follows from step (e) in the proof of Lemma \ref{lem01} and (b) follows from Lemma \ref{lem_pre01}.\\
\indent Also, we have
\begin{equation*}
\begin{split}
tr(\mathbf{K}_{x,n})\leq nP &\Leftrightarrow tr(\mathbf{K}_{s,n}+\mathbf{B}_n(\mathbf{K}_{v,n}+\mathbf{K}_{w,n})\mathbf{B}_n^T)\leq nP \\
&\Leftrightarrow tr(\mathbf{H}_n-\mathbf{K}_{w,n}\mathbf{B}_n^T-\mathbf{B}_n\mathbf{K}_{w,n}-\mathbf{K}_{w,n})\leq nP\\
\end{split}
\end{equation*}

\indent Next, we have the following equivalences by applying the Schur complement. \\
(1).$\det\begin{bmatrix} \mathbf{K}_{v,n}^{-1} & \mathbf{B}_n^T \\ \mathbf{B}_n & \mathbf{H}_n \end{bmatrix}=\det(\mathbf{H}_n-\mathbf{B}_n\mathbf{K}_{v,n}\mathbf{B}_n^T)\det\mathbf{K}_{v,n}^{-1}$.\\
(2).\begin{equation*}
\begin{split}
\mathbf{K}_{s,n}\geq 0 &\Leftrightarrow \mathbf{H}_n-(\mathbf{I}_n+\mathbf{B}_n)\mathbf{K}_{w,n}(\mathbf{I}_n+\mathbf{B}_n)^T-\mathbf{B}_n\mathbf{K}_{v,n}\mathbf{B}_n^T\geq 0\\
&\Leftrightarrow \begin{bmatrix}  \mathbf{H}_n & \mathbf{I}_n+\mathbf{B}_n^T & \mathbf{B}_n^T\\ \mathbf{I}_n+\mathbf{B}_n & \mathbf{K}_{w,n}^{-1}& \mathbf{0}_{n}\\ \mathbf{B}_n & \mathbf{0}_{n}& \mathbf{K}_{v,n}^{-1} \end{bmatrix}\geq 0\\
\end{split}
\end{equation*}

\indent By taking simple replacements on the optimization problem (\ref{eq2_01}), the proof is complete.
\end{proof}

\indent Note that $ \mathbf{H}_n$ is the covariance of the received signal $y^n$. This upper bound provides interesting insight because it shows that the effect of the noise in the feedback link can be formulated as the allocation of the channel input power $P$. As shown in the proof of Theorem \ref{thm01}, we can rewrite (\ref{eq2_01}) as
\begin{equation}
\begin{split}
\quad \underset{\mathbf{B}_n,\mathbf{K}_{s,n}}{\rm maximize} &\quad \frac{1}{2n}\log{\frac{\det{((\mathbf{I}_n+\mathbf{B}_n)\mathbf{K}_{w,n}(\mathbf{I}_n+\mathbf{B}_n)^T+\mathbf{K}_{s,n})}}{\det{\mathbf{K}_{w,n}}}}\\
\text{subject to} &\quad tr(\mathbf{K}_{s,n}+\mathbf{B}_n(\mathbf{K}_{v,n}+\mathbf{K}_{w,n})\mathbf{B}_n^T)\leq nP , \quad \mathbf{K}_{s,n}\geq 0 \\
&\quad \text{$\mathbf{B}_n$ is strictly lower triangular}\\
\end{split}
\label{eq2_02}
\end{equation}
$\mathbf{K}_{v,n}$ herein only affects the power constraint. If $\mathbf{K}_{v,n}=\mathbf{0}_{n}$, the optimization problem (\ref{eq2_02}) recovers the n-block capacity of channels with ideal feedback \cite{cover89}. This implies that, for channels with noisy feedback, it is necessary to assign certain amount of power to cancel the effect of the feedback noise such that the message can be recovered by the decoder with an arbitrarily small error probability. If the noise in the feedback link increases (i.e. $\mathbf{K}_{v,n}$ grows large in some sense), the feedback benefit in increasing reliable transmission rate vanishes. Namely, the noisy feedback system behaves like a nonfeedback system since, due to the power constraint, $\mathbf{B}_n$ approaches $\mathbf{0}_n$ as $\mathbf{K}_{v,n}$ grows. Note that this upper bound is tight when $\mathbf{K}_{v,n}$ is either small or large.

\label{sec_upper}

\section{A Lower Bound on $R_{n,max}^{noisy}(P)$}
\begin{figure}
\includegraphics[scale=0.22]{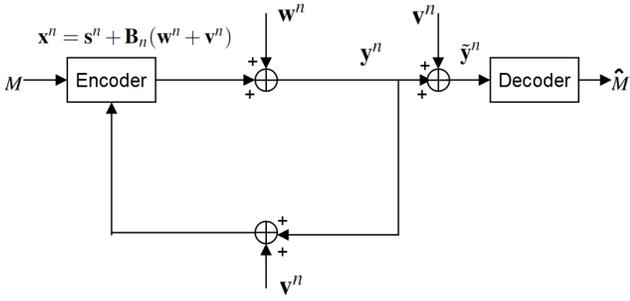}
\caption{A modified additive Gaussian noise channel with noisy feedback}
\label{ISIT_scheme02}
\end{figure}

\indent First of all, we consider a new channel with noisy feedback, as shown in Fig.\ref{ISIT_scheme02}. An identical Gaussian noise $\mathbf{v}$ is added on the channel output. Then,
\begin{equation*}
\begin{split}
\tilde{\mathbf{y}}^n=&\mathbf{x}^n+\mathbf{w}^n+\mathbf{v}^n\\   =&\mathbf{s}^n+(\mathbf{B}_n+\mathbf{I}_n)(\mathbf{w}^n+\mathbf{v}^n)\\
\end{split}
\end{equation*}
Since the decoder is not able to access the new additive noise, the largest achievable rate of the new channel must be a lower bound on that of the original channel (shown in Fig.\ref{ISIT_scheme01}). This motivates us to solve the following optimization problem for obtaining this lower bound.
\begin{equation}
\begin{split}
\quad \underset{\mathbf{B}_n,\mathbf{K}_{s,n}}{\rm maximize} &\quad \frac{1}{n}I(M; \tilde{Y}^n)\\
\text{subject to} &\quad tr(\mathbf{K}_{x,n})\leq nP, \quad \mathbf{K}_{s,n}\geq 0 \\
&\quad \text{$\mathbf{B}_n$ is strictly lower triangular}\\
\end{split}
\label{eq3_01}
\end{equation}

\indent Similarly, we show that the above optimization problem can be transformed into a convex form.
\begin{theorem}
A lower bound on the largest n-block achievable rate of linear feedback coding scheme for Gaussian channels with additive noise feedback, as shown in Fig.\ref{ISIT_scheme01}, can be obtained as the optimal value of the following convex optimization problem.
\begin{equation*}
\begin{split}
\quad \underset{\mathbf{H}_n,\mathbf{B}_{n}}{\rm maximize} &\quad \frac{1}{2n}\log\det\mathbf{H}_n -\frac{1}{2n}\log\det\mathbf{K}_{wv,n}\\
\text{subject to} &\quad tr(\mathbf{H}_n-\mathbf{K}_{wv,n}\mathbf{B}_n^T-\mathbf{B}_n\mathbf{K}_{wv,n}-\mathbf{K}_{wv,n})\leq nP \\
&\quad \begin{bmatrix}  \mathbf{H}_n & \mathbf{I}_n+\mathbf{B}_n^T \\ \mathbf{I}_n+\mathbf{B}_n & \mathbf{K}_{wv,n}^{-1} \end{bmatrix}\geq 0\\
&\quad \text{$\mathbf{B}_n$ is strictly lower triangular}\\
\end{split}
\end{equation*}
where $\mathbf{K}_{wv,n}=\mathbf{K}_{v,n}+\mathbf{K}_{w,n}$.
\label{thm02}
\end{theorem}

\indent The proof is similar to that in \cite{Boyed98} by considering the above setup.
%
%

\begin{remark}
This lower bound is tight when $\mathbf{K}_{v,n}=\mathbf{0}$ and becomes increasing loose as $\mathbf{K}_{v,n}$ increases. This lower bound becomes useless when it is below the corresponding nonfeedback capacity. Since we restrict the feedback coding scheme to be linear, $R_{n,max}^{noisy}$ is in fact a lower bound of the capacity. Therefore, the lower bound of $R_{n,max}^{noisy}$ is obviously valid for the capacity.
\end{remark}

\label{sec_lower}
\section{Simulations and Discussion}
\begin{figure}
\includegraphics[scale=0.53]{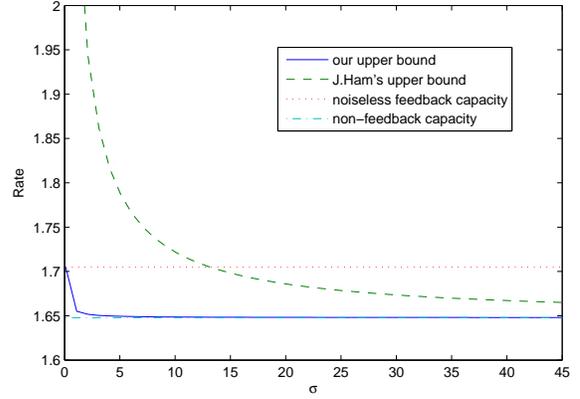}
\caption{Comparison of upper bounds on $C_n^{noisy}$ of the $1$st-MV channel}
\label{sim_comparison}
\end{figure}

\indent In this section, we performed simulations for a specific channel with noisy feedback link. We assumed that the forward channel is created by a first order moving average ($1$st-MV) Gaussian process. Namely,
\begin{equation*}
W_i=U_i+\alpha U_{i-1}
\end{equation*}
where $U_i$ is a white Gaussian process with zero mean and unit variance. We also assumed that the feedback link is created by an additive white Gaussian noise with $\mathbf{K}_{v,n}=\sigma \mathbf{I}_n$ ($\sigma\geq 0$). Because of the practical computation limit, we take coding block length $n=30$ and power limit $P=10$.\\
\indent We first compared our upper bound (i.e. Theorem \ref{thm01}) with the one presented in \cite{Love10} (Lemma $3$). See Fig.\ref{sim_comparison}. As $\sigma$ increases, both of the upper bounds approach the nonfeedback capacity, which implies the ``shut off'' of the feedback link. However, our bound is much more tight, especially, in the small feedback noise region. Note that, when the feedback noise vanishes, the bound in \cite{Love10} on any noisy feedback channel grows to infinity and, thus, should be truncated by the ideal feedback capacity. In contrast, our bound converges to the ideal feedback capacity in this case. Therefore, we may claim that our upper bound is better in general. \\
\indent Next, we computed the bounds derived in our paper for averaging statistic $\alpha=0.3, 0.5,0.9$ in the $1$st-MV channel, as shown in Fig. \ref{sim_fig01_0.1}-\ref{sim_fig01_0.9}. Generally, the plots show that the largest achievable rate $R_{n,max}^{noisy}$, which is in the region between the upper and lower bounds, sharply decreases as $\sigma$ grows. When $\sigma$ grows large enough (i.e. $\sigma=0.8$ in Fig.\ref{sim_fig01_0.1}), the feedback rate-increasing enhancement almost shuts off and, thus, the feedback system behaves like a nonfeedback system. Based on this observation, we conclude that the achievable rate of the additive Gaussian noise channel with noisy feedback is sensitive to the feedback noise under the linear feedback scheme.\\
\indent Additionally, the plots show that the decrease of the achievable rate with $\sigma$ is lesser as $\alpha$ grows. This indicates that the achievable rate is less sensitive to the feedback noise if the channel has more correlated channel noise. This intuitively makes sense since utilize a feedback link for channels with more correlated channel noise would increase more transmission rate and, therefore, the corruption effect of the feedback noise is relatively reduced in this case.

\begin{figure}
\includegraphics[scale=0.46]{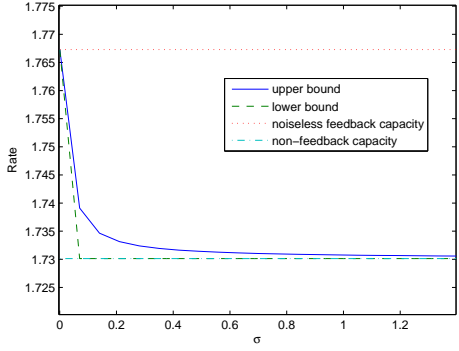}
\caption{The bounds on $C_n^{noisy}$ of the $1$st-MV channel with $\alpha=0.1$}
\label{sim_fig01_0.1}
\end{figure}

\begin{figure}
\includegraphics[scale=0.46]{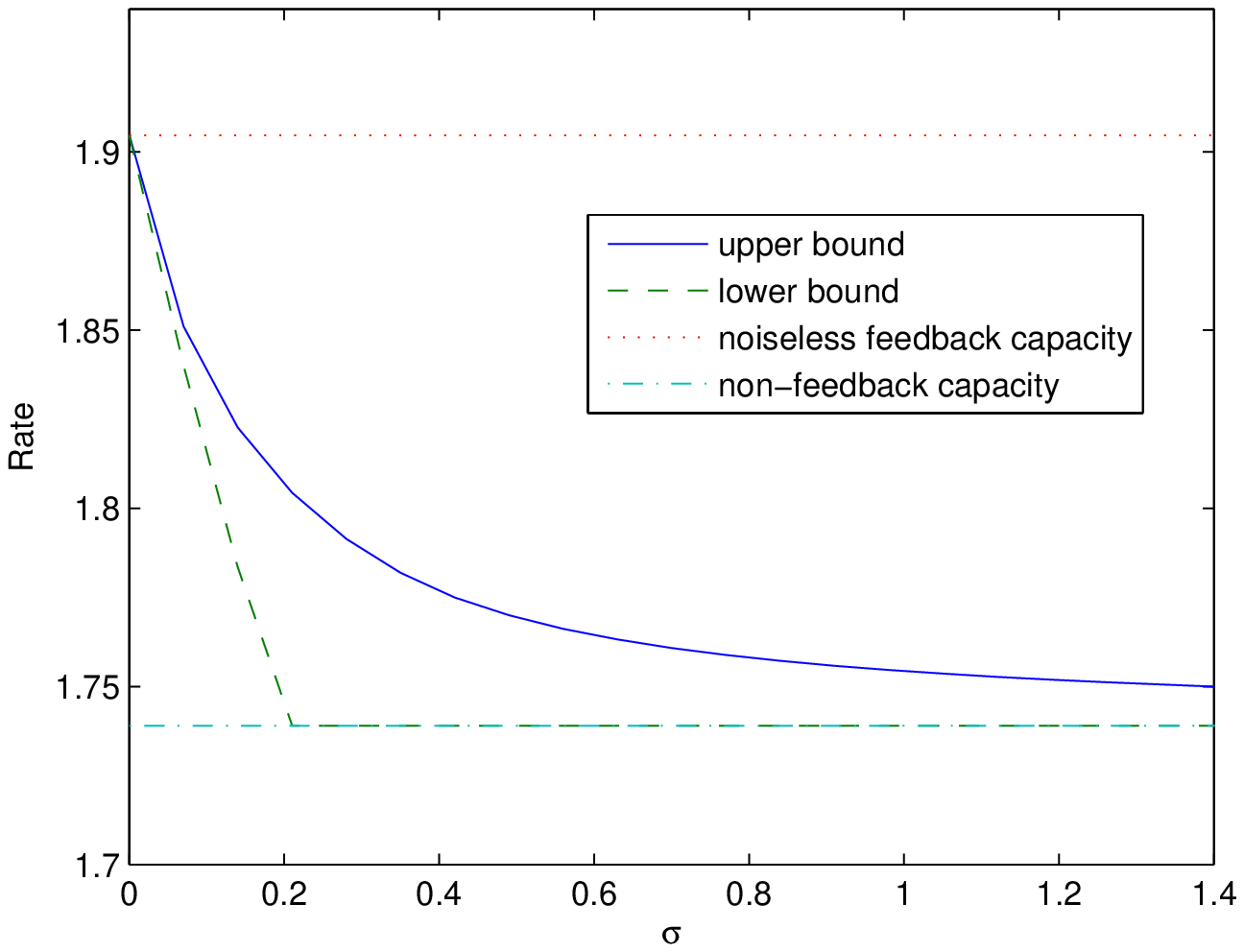}
\caption{The bounds on $C_n^{noisy}$ of the $1$st-MV channel with $\alpha=0.5$}
\label{sim_fig01_0.5}
\end{figure}

\begin{figure}
\includegraphics[scale=0.46]{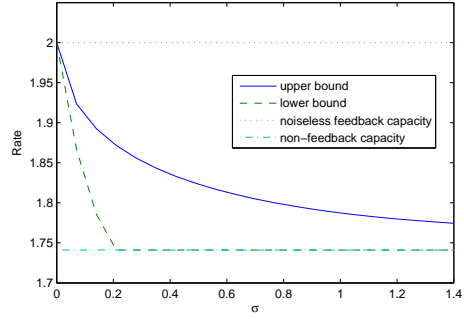}
\caption{The bounds on $C_n^{noisy}$ of the $1$st-MV channel with $\alpha=0.9$}
\label{sim_fig01_0.9}
\end{figure} 
\label{sec_simu}

\section{Conclusion}
\indent We have derived the upper and lower bounds on the largest achievable rate for a linear feedback coding setup. It is shown that these two bounds can be obtained as the optimal values of two convex optimization problems. Furthermore, these bounds provide us the following insight: 1. The achievable rate is very sensitive to the feedback noise. 2. The achievable rate of channels with more correlated channel noise is less sensitive to the feedback noise.
\label{sec_conclude}

\bibliographystyle{IEEEtran}
\bibliography{ref}

\end{document}